\documentclass{article}

\usepackage{amsmath,amsthm,amssymb,amsbsy,epsfig,graphicx,color,multicol}
\usepackage{array,calc}
\usepackage{rotating}
\usepackage{young,epic}
\usepackage{a4wide,bm}
\usepackage{tikz}
\usepackage{cite}
\usepackage{url}
\usepackage{hyperref}
\newcommand{\nc}{\newcommand}
\newtheorem{theorem}{Theorem}[section]

\newtheorem{example}[theorem]{Example}

\nc{\bex}{\begin{example}}
\nc{\eex}{\end{example}}
\nc{\bea}{\begin{eqnarray}}
\nc{\eea}{\end{eqnarray}}
\nc{\ben}{\begin{eqnarray*}}
\nc{\een}{\end{eqnarray*}}
\nc{\bse}{\begin{subequations}}
\nc{\ese}{\end{subequations}}
\nc{\nn}{\nonumber}
\nc{\half}{\ensuremath{\frac{1}{2}}}
\nc{\Hom}{\operatorname{Hom}}
\nc{\End}{\operatorname{End}}
\nc{\vac}{|\textrm{vac}\rangle}
\nc{\tvac}{|\widetilde{\textrm{vac}}\rangle}
\nc{\dvac}{\langle\textrm{vac}}
\nc{\dtvac}{\langle\widetilde{\textrm{vac}}}
\nc{\id}{\mathbb{I}}
\nc{\Tr}{{\rm Tr}}
\nc{\tr}{{\rm Tr}}
\nc{\ree}[1]{\hbox{Re}({#1})}
\nc{\ws}{\;\;}
\newlength\celldim \newlength\fontheight \newlength\extraheight
\newcounter{sqcolumns}
\newcolumntype{S}{ @{}
>{\centering \rule[-0.5\extraheight]{0pt}{\fontheight + \extraheight}}
p{\celldim} @{} }
\newcolumntype{Z}{ @{} >{\centering} p{\celldim} @{} }
{\end{tabular}}

\nc{\ko}{\mbox{$k-1$}}
\nc{\kp}{\mbox{$k+1$}}
\nc{\kth}{\mbox{$k-3$}}
\nc{\kt}{\mbox{$k-2$}}
\nc{\nmo}{\mbox{$n-1$}}
\nc{\nmt}{\mbox{$n-2$}}
\nc{\nmth}{\mbox{$n-3$}}
\nc{\jo}{\mbox{$n-k-h+j+g+1$}}
\nc{\jt}{\mbox{$n-k-h+j+g+2$}}
\nc{\nkj}{\mbox{$n-k+j$}}
\nc{\nkgo}{\mbox{$n-h+g+1$}}
\nc{\nkg}{\mbox{$n-h+g$}}
\nc{\nkjo}{\mbox{$n-k+j+1$}}
\nc{\kho}{\mbox{$k+h-1$}}
\nc{\kht}{\mbox{$k+h-2$}}
\nc{\khth}{\mbox{$k+h-3$}}
\nc{\khf}{\mbox{$k+h-4$}}
\nc{\ho}{\mbox{$h-1$}}
\nc{\hp}{\mbox{$h+1$}}
\nc{\jkho}{\mbox{$j+k+h-1$}}
\nc{\jh}{\mbox{$n-k+j_1$}}
\nc{\jho}{\mbox{$n-k+j-1$}}
\nc{\jk}{\mbox{$n-h+1$}}
\nc{\nht}{\mbox{$n-h+2$}}
\nc{\jko}{\mbox{$n-h+g-1$}}
\nc{\nhgo}{\mbox{$n-h+g+1$}}
\nc{\jkt}{\mbox{$n-h+g-2$}}
\nc{\njhk}{\mbox{$n-k-h+j+g$}}
\nc{\mkh}{\mbox{Max$(k,h)$}}
\nc{\uo}{\mbox{$u+1$}}
\nc{\ut}{\mbox{$u+2$}}
\nc{\uth}{\mbox{$u+3$}}
\nc{\uno}{\mbox{$n+u-1$}}
\nc{\unt}{\mbox{$n+u-2$}}
\nc{\unkj}{\mbox{$n-k+u+j$}}
\nc{\unkjo}{\mbox{$n-k+u+j-1$}}
\nc{\nkog}{\mbox{$n-k-1+g$}}
\nc{\nhg}{\mbox{$n-k+g$}}
\nc{\nhjh}{\mbox{$n-h+1+j_h$}}
\nc{\nhjho}{\mbox{$n-k-h+2+j_{h-1}$}}
\nc{\nhjt}{\mbox{$n-k-1+j_{2}$}}
\nc{\nkhjo}{\mbox{$n-k+j_{1}+1$}}
\nc{\e}{\mbox{e}}
\nc{\ga}{\alpha}
\nc{\gb}{\beta}
\nc{\gd}{\delta}
\nc{\gep}{\varepsilon}
\nc{\gz}{\zeta}
\nc{\gt}{\theta}
\nc{\gk}{\kappa}
\nc{\gl}{\lambda}
\nc{\gp}{\phi}
\nc{\gs}{\sigma}
\nc{\go}{\omega}
\nc{\gn}{\nu}
\nc{\gr}{\rho}
\nc{\gm}{\mu}

\nc{\gou}{\underline{\go}}
\nc{\un}{\underline{n}}
\nc{\um}{\underline{m}}
\nc{\uw}{\underline{\go}}

\nc{\s}{\sigma}
\nc{\ep}{\varepsilon}
\nc{\z}{\zeta}
\nc{\g}{\gamma}
\nc{\zi}{\zeta^{-1}}

\nc{\gG}{\Gamma}
\nc{\gD}{\Delta}
\nc{\gT}{\Theta}
\nc{\gL}{\Lambda}
\nc{\gO}{\Omega}
\nc{\gP}{\Phi}

\nc{\cF}{\mathcal{F}}
\nc{\cP}{\mathcal{P}}
\nc{\cS}{\mathcal{S}}
\nc{\cN}{\mathcal{N}}
\nc{\cD}{\mathcal{D}}
\nc{\cH}{\mathcal{H}}
\nc{\cO}{\mathcal{O}}
\nc{\cT}{\mathcal{T}}
\nc{\cQ}{\mathcal{Q}}
\nc{\cW}{\mathcal{W}}
\nc{\cR}{\mathcal{R}}
\nc{\cC}{\mathcal{C}}

\nc{\pt}{\mathcal{P}\mathcal{T}}

\nc{\C}{\mathbb{C}}
\nc{\Q}{\mathbb{Q}}
\nc{\R}{\mathbb{R}}
\nc{\Z}{\mathbb{Z}}
\nc{\N}{\mathbb{N}}

\nc{\fg}{\mathfrak{g}}

\nc{\barx}{\bar{x}}
\nc{\bi}{\bar{i}}
\nc{\bj}{\bar{j}}
\nc{\bgr}{\bar{\rho}}
\nc{\bA}{\bar{\alpha}}
\nc{\bB}{\bar{\beta}}
\nc{\bC}{\bar{\gamma}}
\nc{\by}{\bar{y}}
\nc{\brv}{\overline{V}}
\nc{\brp}{\overline{P}}
\nc{\T}{\tilde{T}}
\nc{\tf}{\tilde{f}}
\nc{\te}{\tilde{e}}
\nc{\ts}{\tilde{s}}
\nc{\tgP}{\widetilde{\Phi}}
\nc{\tgPs}{\tilde{\Psi}}
\nc{\tgn}{\tilde{\nu}}
\nc{\tgl}{\tilde{\lambda}}
\nc{\tge}{\tilde{\eta}}
\nc{\txi}{\tilde{\xi}}
\nc{\tep}{\tilde{\epsilon}}
\nc{\tx}{\tilde{x}}

\nc{\cB}{\check{b}}
\nc{\cOm}{\check{\Omega}}

\nc{\goto}{\mapsto}
\DeclareMathSymbol{\Gamma}{\mathalpha}{letters}{"00}
\DeclareMathSymbol{\Delta}{\mathalpha}{letters}{"01}
\DeclareMathSymbol{\Theta}{\mathalpha}{letters}{"02}
\DeclareMathSymbol{\Lambda}{\mathalpha}{letters}{"03}
\DeclareMathSymbol{\Xi}{\mathalpha}{letters}{"04}
\DeclareMathSymbol{\Pi}{\mathalpha}{letters}{"05}
\DeclareMathSymbol{\Sigma}{\mathalpha}{letters}{"06}
\DeclareMathSymbol{\Upsilon}{\mathalpha}{letters}{"07}
\DeclareMathSymbol{\Phi}{\mathalpha}{letters}{"08}
\DeclareMathSymbol{\Psi}{\mathalpha}{letters}{"09}
\DeclareMathSymbol{\Omega}{\mathalpha}{letters}{"0A}

\newcommand{\dd}{\mathrm{d}}
\newcommand{\ii}{\mathrm{i}}

\newcommand{\be}{\begin{equation}}
\newcommand{\ee}{\end{equation}}

\newcommand{\state}[1]{\mathopen{|}#1\mathclose{\rangle}}

\newcommand{\ket}[1]{\left|#1\right\rangle}      % Ket-Zustand
\newcommand{\bra}[1]{\left\langle #1\right|}     % Bra-Zustand

\newcommand{\alg}[1]{\mathfrak{#1}}

\makeatletter
\def\mr@ignsp#1 {\ifx\:#1\@empty\else #1\expandafter\mr@ignsp\fi}%
\newcommand{\multiref}[1]{\begingroup%\let\protect\string%
\xdef\mr@no@sparg{\expandafter\mr@ignsp#1 \: }%
\def\mr@comma{}%
\@for\mr@refs:=\mr@no@sparg\do{\mr@comma\def\mr@comma{,}\ref{\mr@refs}}%
\endgroup}
\makeatother

\newcommand{\hypref}[2]{\ifx\href\asklfhas #2\else\href{#1}{#2}\fi}
\newcommand{\Secref}[1]{Section~\multiref{#1}}
\newcommand{\secref}[1]{Sec.~\multiref{#1}}

\newcommand{\figref}[1]{Fig.~\multiref{#1}}
\renewcommand{\eqref}[1]{(\multiref{#1})}

\begin{document}
\title{Multiple integral formula for the off-shell six vertex scalar product}
\author{Jan de Gier$^1$, Wellington Galleas$^2$ and Mark Sorrell$^{1,2}$\\[5mm]
\textit{
\begin{minipage}{0.9\textwidth}\small
${}^1$Department of Mathematics and Statistics, The University of Melbourne, VIC 3010, Australia\\ 
\end{minipage}}\\
\textit{
\begin{minipage}{0.9\textwidth}\small
${}^2$ARC Centre of Excellence for the Mathematics and Statistics of Complex Systems, The University of Melbourne, VIC 3010, Australia\\ 
\end{minipage}}
}

\maketitle

\begin{abstract}
We write a multiple integral formula for the partition function of the Z-invariant six vertex model and 
demonstrate how it can be specialised to compute the norm of Bethe vectors. We also discuss the possibility of computing three-point functions in $\mathcal{N}=4$ SYM using these integral formul\ae. 
\end{abstract}

\section{Introduction}
The six vertex model is one of the most studied solvable lattice model in statistical mechanics \cite{Baxter82}. Using the Bethe ansatz for eigenvectors and eigenvalues of the transfer matrix, Lieb \cite{Lieb67_Oct,Lieb67_April} and Sutherland \cite{Sutherland67} solved the regular square lattice zero-field six-vertex model with periodic boundary conditions. By computing the largest eigenvalue they obtained explicit solutions for the free-energy per site in the thermodynamic limit. 

The Bethe ansatz solution for a periodic square lattice with $L^2$ sites requires that for each eigenvalue one has to solve $n$ non-linear
equations in $n$ unknowns (the `wave numbers' $k_1,\ldots,k_n$). This can only be done explicitly 
in the large-$L$ limit under the `string hypothesis'. Usually one can solve those equations explicitly only for the largest and near-largest eigenvalues, which limits one to considering the full infinite square lattice. 

On the infinite square lattice vertex operator methods are available (see e.g. \cite{Jimbo_Miwa94}), which have 
resulted in multiple integral expression for correlation functions of the model.

The situation is different when one considers fixed boundary conditions. Izergin and Korepin constructed an explicit determinant 
solution for the case of `domain wall boundary conditions' \cite{Izergin87,Korepin82}. Alternative representations for this same partition function have also been proposed by other authors. See for instance 
\cite{Stroganov04,Stroganov06,Khoro_Paku05,Lascoux07,Galleas10,Galleas11} and references therein.

In \cite{Baxter87}, Baxter constructs the partition function for a general $Z$-invariant six-vertex model, which is given 
explicitly (apart from a simple factor) by a Bethe ansatz type expression. However, in this case the `wave numbers' do not need to
be evaluated from a complicated set of simultaneous equations: they are known explicitly. The disadvantage of Baxter's expression however is that it is given in terms of a large summation over the symmetric group.

The main results of this paper is to rewrite Baxter's expression as a multiple contour integral over a factorised polynomial kernel. As an 
application of this result we specialise it to the regular square lattice with two specific boundary conditions. 
This allows us to write a new representation for the partition function of the six vertex model with domain wall boundary conditions
as a multiple contour integral. Likewise we derive an off-shell multiple integral expression for the scalar product of two Bethe vectors 
of the six-vertex model transfer matrix. 

Such formul\ae are important as recent developments in the computation 
of three-point functions in $\mathcal{N}=4$ Super Yang-Mills (SYM) seem to require manageable expressions for the norms of Bethe vectors for general values of the parameters. With this in mind, in Section~\ref{sec:3pt} we illustrate how the integral formul\ae presented here can be embedded in the framework of \cite{Escobedo_2011}. 

\section{The six vertex model}
\label{sec:vertexmodel}

The exact solvability of a vertex model in the sense of Baxter \cite{Baxter_book} is intimately related to the solutions of the Yang-Baxter equation. This equation reads 
\be 
\label{YB}
\mathcal{L}_{12} (\lambda - \mu) \mathcal{L}_{13} (\lambda) \mathcal{L}_{23} (\mu) = \mathcal{L}_{23} (\mu) \mathcal{L}_{13} (\lambda) \mathcal{L}_{12} (\lambda - \mu) 
\ee
where $\mathcal{L}_{ij} \in \mbox{End}( \mathbb{V}_i  \otimes \mathbb{V}_j )$. The complex variables $\lambda$ and $\mu$ correspond to spectral parameters while $\mathbb{V}_i$ denotes a complex vector space. We refer to the solutions of (\ref{YB}) as $\mathcal{L}$-matrices and for the six vertex model the corresponding $\mathcal{L}$-matrix is invariant under the $U_{q}[\widehat{\alg{sl}}(2)]$ algebra in the fundamental representation. In that case $\mathbb{V}_i \simeq \mathbb{C}^2$ and the associated $\mathcal{L}$-matrix explicitly reads
\be
\label{su2}
\mathcal{L} = 
\left(
\begin{array}{cccc}
 a & 0 & 0 & 0 \\
 0 & b & c & 0 \\
 0 & c & b & 0 \\
 0 & 0 & 0 & a 
\end{array}
\right)
\ee
where 
\be
\label{weightsdef}
a (\lambda) = \sinh{(\lambda+\gamma)},\quad b (\lambda) = \sinh{(\lambda)}\quad \text{and}\quad c(\lambda) = \sinh{(\gamma)}.
\ee

To characterise the vertex model associated with a solution of the Yang-Baxter equation we shall employ the following notation,
\be
\label{decomp}
\mathcal{L} = \sum_{\alpha,\beta,\gamma,\delta} \mathcal{L}^{\beta \delta}_{\alpha \gamma} \; \hat{e}_{\alpha \beta} \otimes \hat{e}_{\gamma \delta} \; ,
\ee
where 
$\hat{e}_{\alpha \beta} = \ket{\alpha} \bra{\beta}$ where the vectors $\ket{\alpha}$, as well as their duals, form a basis of $\mathbb{V}_i$. In this way  we associate the Boltzmann weight $\mathcal{L}^{\beta \delta}_{\alpha \gamma}$ to the vertex configuration $\{\alpha, \beta, \gamma, \delta \}$ as shown in the \figref{fig:bw}.
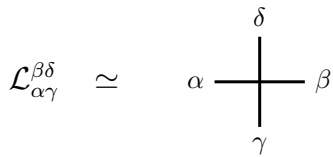
\begin{figure}[h]\centering
\begin{tikzpicture}[scale=0.6,line width=0.4mm]
\draw (-2.3,0) node[left] {{\large $\mathcal{L}^{\beta \delta}_{\alpha \gamma} \;\;\; \simeq \;\;\;$}};
\draw (-1,0) node [left] {$\alpha$} -- (1,0) node[right]{$\beta$};
\draw (0,-1) node[below]{$\gamma$} -- (0,1) node[above]{$\delta$};
\end{tikzpicture}
\caption{Vertex $\{\alpha, \beta, \gamma, \delta \}$ and its Boltzmann weight.}
\label{fig:bw}
\end{figure}

For the six vertex model the indices $\alpha$ and $\beta$ in (\ref{decomp}) run through the set $\{\uparrow , \downarrow \}$ while $\gamma$ and $\delta$ run through $\{ \rightarrow , \leftarrow \}$ 
defined as
\be
\state{\uparrow} = \state{\rightarrow} = \left( \begin{matrix}
1 \cr
0 \end{matrix} \right) \in \mathbb{V}_i \qquad \mbox{and} \qquad
\state{\downarrow} = \state{\leftarrow} = \left( \begin{matrix}
0 \cr
1 \end{matrix} \right) \in \mathbb{V}_i  \; .
\ee
From (\ref{su2}) we thus have six possible vertex configurations which are depicted in the \figref{fig:vertex6}.
\begin{figure}[h]\centering
\begin{tikzpicture}[scale=0.5,line width=0.4mm]
\draw[>->] (-1,0) -- (1,0) ;
\draw[>->] (0,-1) node[below]{$a$} -- (0,1);
\draw[<-<] (4,0) -- (6,0);
\draw[<-<] (5,-1) node[below]{$a$} -- (5,1);
\draw[>->] (9,0) -- (11,0) ;
\draw[<-<] (10,-1) node[below]{$b$} -- (10,1);
\draw[<-<] (14,0) -- (16,0) ;
\draw[>->] (15,-1) node[below]{$b$} -- (15,1);
\draw[>-<] (19,0) -- (21,0);
\draw[<->] (20,-1) node[below]{$c$} -- (20,1);
\draw[<->] (25,0) -- (27,0);
\draw[>-<] (26,-1) node[below]{$c$} -- (26,1);
\end{tikzpicture}
\caption{The 6 vertex configurations of the model.}
\label{fig:vertex6}
\end{figure}
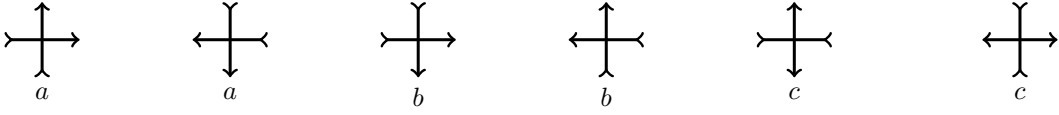
The vertices described in \figref{fig:vertex6} can be juxtaposed in a two-dimensional lattice with $N$ rows and $L$ columns as shown in the \figref{fig:lattice}. 
\begin{figure}[h]\centering
\begin{tikzpicture}[scale=0.5,line width=0.4mm]
{\scriptsize
\draw (-2,0) node[left]{$\alpha_{N,1}$} -- (10,0) node[right]{$\alpha_{N,L+1}$};
\draw (0,-2) -- (0,2) node[above]{$\beta_{N+1,1}$};
\draw (2,-2) -- (2,2);
\draw (3,2) node[above]{\dots};
\draw (4,-2) -- (4,2);
\draw (5,2) node[above]{\dots};
\draw (6,-2) -- (6,2);
\draw (8,-2) -- (8,2) node[above]{$\;\;\;\;\; \beta_{N+1,L}$};
\begin{scope}[yshift=-2.2cm]
\draw (-2,0) node[left]{\vdots} -- (10,0) node[right]{\vdots};
\draw (0,-2) -- (0,2);
\draw (2,-2) -- (2,2);
\draw (4,-2) -- (4,2);
\draw (6,-2) -- (6,2);
\draw (8,-2) -- (8,2);
\end{scope}
\begin{scope}[yshift=-4.4cm]
\draw (-2,0) node[left]{\vdots} -- (10,0) node[right]{\vdots};
\draw (0,-2) -- (0,2);
\draw (2,-2) -- (2,2);
\draw (4,-2) -- (4,2);
\draw (6,-2) -- (6,2);
\draw (8,-2) -- (8,2);
\end{scope}
\begin{scope}[yshift=-6.6cm]
\draw (-2,0) node[left]{$\alpha_{1,1}$} -- (10,0) node[right]{$\alpha_{1,L+1}$};
\draw (0,-2) node[below]{$\beta_{1,1}$} -- (0,2);
\draw (2,-2) -- (2,2);
\draw (3,-2) node[below]{\dots};
\draw (4,-2) -- (4,2);
\draw (5,-2) node[below]{\dots};
\draw (6,-2) -- (6,2);
\draw (8,-2) -- (8,2);
\draw (8,-2) node[below]{$\beta_{1,L}$} -- (8,2);
\end{scope}}
\end{tikzpicture}
\caption{Two dimensional square lattice with boundary conditions index by $\alpha_{i,j}$ and $\beta_{i,j}$ who take values in $\{\leftarrow,\rightarrow\}$ and $\{\uparrow,\downarrow\}$ respectively.}
\label{fig:lattice}
\end{figure}
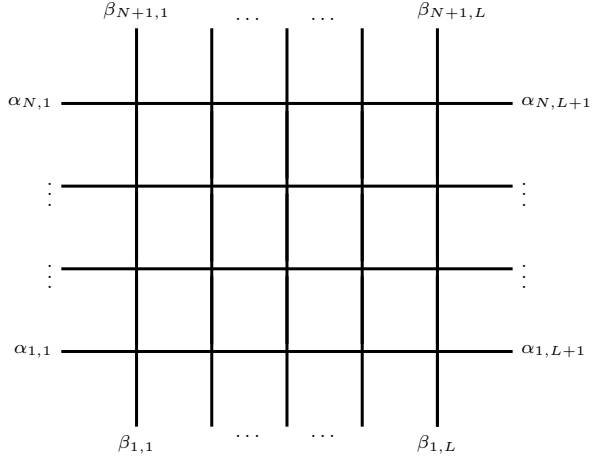
The probability of having the vertex $\{ \alpha_{i,j}, \alpha_{i,j+1}, \beta_{i,j} , \beta_{i+1,j} \}$ at the 
Cartesian coordinates $(i,j)$ is weighted by the factor $\mathcal{L}_{\alpha_{i,j} \beta_{i,j}}^{\alpha_{i,j+1} \beta_{i+1,j} }$ and the partition function of the system is formed by the summation over all possible configurations of the product of vertex weights. More precisely, the partition function $Z$ is defined by
\be
\label{partition}
Z = \sum_{\{ \alpha_{i,j} , \beta_{i,j} \}} \prod_{i=1}^{N} \prod_{j=1}^{L} \mathcal{L}_{\alpha_{i,j} \beta_{i,j}}^{\alpha_{i,j+1} \beta_{i+1,j} } \; .
\ee

It is possible to generalise this construction in the following way to include inhomogeneities while maintaining integrability. 
To each horizontal line $i$ of the square lattice we associate a variable $\lambda_i$, and to each vertical $j$ a variable $\mu_j$. 
The local Boltzmann weight for a vertex $\{ \alpha_{i,j}, \alpha_{i,j+1}, \beta_{i,j} , \beta_{i+1,j} \}$  is then given by $\mathcal{L}_{\alpha_{i,j} \beta_{i,j}}^{\alpha_{i,j+1} \beta_{i+1,j} }(\lambda_i-\mu_j)$
and the inhomogeneous partition function is defined as 
\be
\label{partition2}
Z(\{\lambda_i\},\{\mu_j\}) = \sum_{\{ \alpha_{i,j} , \beta_{i,j} \}} \prod_{i=1}^{N} \prod_{j=1}^{L} \mathcal{L}_{\alpha_{i,j} \beta_{i,j}}^{\alpha_{i,j+1} \beta_{i+1,j} } (\lambda_i-\mu_j) \; .
\ee

Boundary conditions play a fundamental role in the evaluation of the partition function \eqref{partition}. In this paper we will consider 
periodic boundary conditions in the horizontal and vertical directions ($\alpha_{i,L+1}=\alpha_{i,1}$ and $\beta_{N+1,j}=\beta_{j,1}$) 
as well as fixed boundary conditions with specified values of $\alpha_{i,j}$ and $\beta_{i,j}$ at the borders.

\subsection{Monodromy and transfer matrix.} Each row of the two-dimensional lattice formed by the juxtaposition of vertices can be conveniently characterised by a matrix $\mathcal{T} (\lambda) = \mathcal{T} (\lambda , \{ \mu_j \})$ usually referred to as monodromy matrix. This matrix has components
\be
\label{eq:mono1}
\mathcal{T}^{\alpha', \beta'_1 \dots \beta'_L}_{\alpha, \beta_1 \dots \beta_L} (\lambda , \{ \mu_j \}) = \mathcal{L}^{\alpha_1 \beta'_1}_{\alpha \; \beta_1}  (\lambda - \mu_1)  \mathcal{L}^{\alpha_2 \beta'_2}_{\alpha_1 \beta_2} (\lambda - \mu_2) \dots \mathcal{L}^{\alpha_L \beta'_{L-1}}_{\alpha_{L-1} \beta_{L-1}} (\lambda - \mu_{L-1}) \mathcal{L}^{\alpha' \; \beta'_L}_{\alpha_L \beta_L} (\lambda - \mu_L)
\ee
which is diagrammatically represented in the \figref{fig:mono}. 
\begin{figure}[h]\centering
\begin{tikzpicture}[scale=0.4,line width=0.4mm]
\draw (-2.3,0) node[left]{{\large $\mathcal{T}^{\alpha', \beta'_1 \dots \beta'_L}_{\alpha, \beta_1 \dots \beta_L} \;\;\; \simeq \;\;\;$}};
\draw (-2,0) node[left]{$\alpha$} -- (14,0) node[right]{$\alpha'$};
\draw (0,-2) node[below]{$\beta_1$} -- (0,2) node[above]{$\beta'_1$};
\draw (2,-2) node[below]{$\beta_2$} -- (2,2) node[above]{$\beta'_2$};
\draw (4,-2) -- (4,2);
\draw (5,-2) node{\dots};
\draw (5,2) node{\dots};
\draw (6,-2) -- (6,2);
\draw (7,-2) node{\dots};
\draw (7,2) node{\dots};
\draw (8,-2) -- (8,2);
\draw (10,-2) node[below]{$\beta_{L-1}$} -- (10,2) node[above]{$\beta'_{L-1}$};
\draw (12,-2) node[below]{$\beta_L$} -- (12,2) node[above]{$\beta'_L$};
\end{tikzpicture}
\caption{Diagrammatic representation of the monodromy matrix.}
\label{fig:mono}
\end{figure}
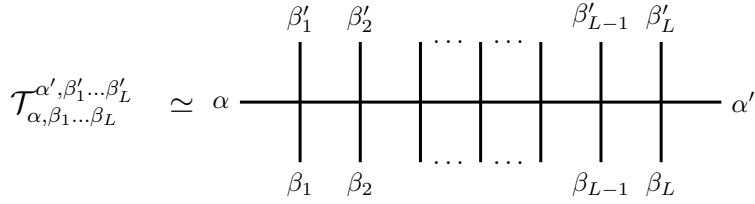
In this way the two-dimensional lattice described in the \figref{fig:lattice} can be completely built in terms of the operators $\mathcal{T} (\lambda , \{ \mu_j \})$.

In a more compact form we have
\be
\label{mono}
\mathcal{T}(\lambda, \{ \mu_j \}) = \mathcal{L}_{\mathcal{A} 1} (\lambda - \mu_1) \dots \mathcal{L}_{\mathcal{A} L} (\lambda - \mu_L).
\ee
which is an operator on the tensor product space $\mathbb{V}_{\mathcal{A}} \otimes \mathbb{V}_{1} \otimes \dots \otimes  \mathbb{V}_{L}$. The space $\mathbb{V}_{\mathcal{A}}$ is usually called the auxiliary space while we shall refer to the tensor product $\mathbb{V}_{1} \otimes \dots \otimes  \mathbb{V}_{L}$ as the quantum space. Thus the monodromy matrix $\mathcal{T}(\lambda)$ is a matrix in the auxiliary space whose elements $\mathcal{T}^{\alpha'}_{\alpha}$ act on the quantum space. In the diagrammatic representation of the monodromy matrix given in \figref{fig:mono}, the contraction of indices labelling the horizontal edges represent the product in the auxiliary space while the product in the quantum space is given by the contraction of indices labelling the vertical edges. For the six vertex model the monodromy matrix can be conveniently denoted as
\be
\label{monorep}
\mathcal{T} (\lambda) = \mathcal{T}(\lambda, \{ \mu_j \}) = \left( \begin{matrix}
A(\lambda) &  B (\lambda) \cr
C (\lambda) & D (\lambda) \end{matrix}  
\right) \; .
\ee
Here and in the following we suppress where possible the dependence on $\{\mu_j\}$ for clarity of notation.

\section{Periodic boundary conditions}
The evaluation of the partition function (\ref{partition}) depends drastically on the boundary conditions chosen. In this
section we shall consider periodic, or toroidal boundary conditions ($\alpha_{i,L+1}=\alpha_{i,1}$ and $\beta_{N+1,j}=\beta_{j,1}$), and review the basics of the quantum inverse scattering method (QISM) for computing the six-vertex partition function in this case.

From \eqref{partition} and \eqref{eq:mono1} it is a standard construction that the partition function for toroidal boundary conditions can be written as
\be
Z(\{\lambda_i\},\{\mu_j\}) = \Tr_{\mathbb{C}^{2\otimes L}} \left(\tau(\lambda_1)\cdots \tau(\lambda_{N}) \right),
\ee
where $\tau(\lambda)$ is the transfer matrix defined by
\be
\label{transfmatrix}
\tau(\lambda) = \tau(\lambda,\{\mu_j\}) = \Tr_{\mathcal{A}} \mathcal{T}(\lambda) = A(\lambda)+D(\lambda).
\ee
Due to the Yang-Baxter equation the transfer matrices $\tau(\lambda_i)$ all commute and are simultaneously diagonalisable. The partition function is therefore given by
\be
Z(\{\lambda_i\},\{\mu_j\}) = \sum_{\ell=1}^{2^L} \prod_{i=1}^N \Lambda_{\ell}(\lambda_i),
\ee
where $\Lambda_{\ell}$ is the $\ell$th eigenvalue of the transfer matrix $\tau$. 

\subsection{Bethe ansatz diagonalisation}

Let us first define the `pseudo-vacuum' states $\ket{\Psi_{\uparrow}}$ and $\ket{\Psi_{\downarrow}}$ by
\be
\label{states}
\ket{\Psi_\uparrow} = \bigotimes_{i=1}^{L} \ket{\uparrow} \qquad \qquad \mbox{and} \qquad \qquad \ket{\Psi_\downarrow} = \bigotimes_{i=1}^{L} \ket{\downarrow} \; .
\ee
An important ingredient in the QISM setup for diagonalising the transfer matrix $\tau(\lambda)$ is the action of the monodromy matrix  elements on the states $\ket{\Psi_{\uparrow}}$ and $\ket{\Psi_{\downarrow}}$. For $\ket{\Psi_{\uparrow}}$ these are given by
\begin{align}
\label{action}
A(\lambda) \ket{\Psi_{\uparrow}} &= \prod_{j=1}^L a(\lambda - \mu_j) \ket{\Psi_{\uparrow}} & B(\lambda) \ket{\Psi_{\uparrow}} & = * \nonumber \\
C(\lambda) \ket{\Psi_{\uparrow}} &= 0 &    D(\lambda) \ket{\Psi_{\uparrow}} &= \prod_{j=1}^L b(\lambda - \mu_j) \ket{\Psi_{\uparrow}}\; .
\end{align}
%
%
%\begin{align}
%\label{action1}
%A(\lambda) \ket{\Psi_{-}} &= \prod_{j=1}^L b(\lambda - \mu_j) \ket{\Psi_{-}}&   B(\lambda) \ket{\Psi_{-}} = & 0 \nonumber \\
%C(\lambda) \ket{\Psi_{-}} &= *&    D(\lambda) \ket{\Psi_{-}} = & \prod_{j=1}^L a(\lambda - \mu_j) \ket{\Psi_{-}} \; .
%\end{align}
%
The relations (\ref{action}) %and (\ref{action1}) 
follow from the definitions (\ref{mono}) and (\ref{monorep}) together with the triangularity of $\mathcal{L}_{\mathcal{A} j}$ on local states.

Due to the Yang-Baxter equation (\ref{YB}) one can readily show that the monodromy matrix (\ref{mono})
satisfies the following quadratic relation
\be 
\label{yba}
\mathcal{R}(\lambda - \nu) \left[\mathcal{T}(\lambda) \otimes \mathcal{T}(\nu)\right] = \left[\mathcal{T}(\nu) \otimes \mathcal{T}(\lambda)\right] \mathcal{R}(\lambda - \nu)
\ee
which is known commonly as the Yang-Baxter algebra. Here $\mathcal{R} = P \mathcal{L}$ where $P$ stands for the standard permutation matrix. Thus (\ref{yba}) yields commutation relations for the elements of the monodromy matrix (\ref{monorep}). Among the relations encoded in (\ref{yba}) we shall make use of the following ones,
\begin{align}
A(\lambda) B(\nu) &= \frac{a(\nu - \lambda)}{b(\nu - \lambda )} B(\nu) A(\lambda) - 
\frac{c(\nu - \lambda)}{b(\nu - \lambda )} B(\lambda) A(\nu), \nonumber \\
D(\lambda ) B(\nu) &= \frac{a(\lambda - \nu)}{b(\lambda - \nu)} B(\nu ) D(\lambda ) - 
\frac{c(\lambda - \nu)}{b(\lambda - \nu )} B(\lambda) D(\nu ), \label{alg}\\ 
B(\lambda)  B(\nu ) &=  B(\nu ) B(\lambda )   \; .\nonumber 
\end{align}
We are now in a position to state the main theorem of the algebraic Bethe ansatz, or quantum inverse scattering method.
\begin{theorem}
Let the numbers $\lambda_1,\ldots,\lambda_n$ be solutions of the equations
\be
\label{eq:bae}
\prod_{j=1}^L \frac{a(\lambda_i - \mu_j)}{b(\lambda_i - \mu_j) } = (-1)^{n-1}\prod_{k=1}^n \frac{a(\lambda_i - \lambda_k)}{a(\lambda_k - \lambda_i) }.
\ee
Then, 
\be
\label{eq:BAeigvec}
\ket{\lambda_1,\ldots,\lambda_n} = B(\lambda_1)\cdots B(\lambda_n) \ket{\Psi_{\uparrow}}
\ee
are eigenvectors of the transfer matrix $\tau(\lambda,\{\mu_j\})$
with eigenvalues $\Lambda(\lambda,\{\mu_j\})$ given by 
\begin{align}
\label{eq:BAeigval}
\Lambda(\lambda,\{\mu_j\}) &= \prod_{j=1}^L a(\lambda - \mu_j) \prod_{k=1}^n \frac{a(\lambda_k-\lambda)}{b(\lambda_k-\lambda)} + \prod_{j=1}^L b(\lambda - \mu_j) \prod_{k=1}^n \frac{a(\lambda-\lambda_k)}{b(\lambda-\lambda_k)} \; .
% \label{eq:BAeigvec}
% \ket{\lambda_1,\ldots,\lambda_n} = B(\lambda_1)\cdots B(\lambda_n) \ket{\Psi_{\uparrow}} 
\end{align}
\end{theorem}
\begin{proof}
The theorem follows from the definition \eqref{transfmatrix} of the transfer matrix, the action \eqref{action} on the pseudo-vacuum, and relations \eqref{alg}. In particular, the expression for the eigenvalue originates from \eqref{action} and the first terms on the right hand side of \eqref{alg}. Equations \eqref{eq:bae} imply that unwanted terms arising from the other terms on the right hand side of \eqref{alg} cancel.
\end{proof}

An important problem in the theory of solvable lattice models is to provide manageable expressions for correlation functions. These are of the form 
\be
\bra{\nu_1,\ldots,\nu_m} \mathcal{O} \ket{\lambda_1,\ldots,\lambda_n}, 
\ee
where $\mathcal{O}$ is some operator. The simplest case concerns an expression for the norm of Bethe states,
\be
\label{eq:scprod}
N(\{\nu_i\},\{\lambda_j \}) = \bra{\nu_1,\ldots,\nu_n} \lambda_1,\ldots,\lambda_n\rangle.
\ee
Below we shall see that \eqref{eq:scprod} is equal to a special case of the partition function \eqref{partition2} with fixed boundary conditions. In the next section we show that such partition functions can be realised as multiple contour integrals over a factorised kernel. 

\section{Fixed boundary conditions}
In this section we review some known results for the partition function \eqref{partition2} in the case of fixed boundary conditions. In particular we recall Baxter's formula for such partition functions for general graphs and boundary conditions, expressed in terms of a sum over the symmetric group \cite{Baxter87}. \\ Let us first look at the well known case of domain wall boundary conditions.

%%%%%%%%%%%%%%%%%%%%%%%%%%%%%%%%%%%%%%%%%%%%%%%%%%%%%%%%%%%%%%%%%%%%%%%%%%%%%%%%
\subsection{Domain wall boundary conditions}
\label{sec:altern}

For the case of domain wall boundaries with $N=L$, the edges on the vertical and horizontal borders in Figure \ref{fig:lattice} assume a particular configuration. Naturally one needs to be careful when choosing the boundaries since an inappropriate choice can render a trivial (null) partition function or force the system to freeze in a particular configuration. Domain wall boundaries for the six vertex model were first introduced by Korepin \cite{Korepin82}, and in our notation correspond to the choice $\alpha_{i,1} = \rightarrow$, $\beta_{L+1,j} = \uparrow$ and $\alpha_{i,L+1} = \leftarrow$, $\beta_{1,j} = \downarrow$. This special boundary condition is illustrated in \figref{fig:latticedw}.
\begin{figure}[h]\centering
\begin{tikzpicture}[scale=0.5,line width=0.4mm]
\draw[>-<] (-2,0) -- (10,0);
\draw[->] (0,-2) -- (0,2);
\draw[->] (2,-2) -- (2,2);
\draw[->] (4,-2) -- (4,2);
\draw[->] (6,-2) -- (6,2);
\draw[->] (8,-2) -- (8,2);
\begin{scope}[yshift=-2.2cm]
\draw[>-<] (-2,0) -- (10,0);
\draw (0,-2) -- (0,2);
\draw (2,-2) -- (2,2);
\draw (4,-2) -- (4,2);
\draw (6,-2) -- (6,2);
\draw (8,-2) -- (8,2);
\end{scope}
\begin{scope}[yshift=-4.4cm]
\draw[>-<]  (-2,0) -- (10,0);
\draw (0,-2) -- (0,2);
\draw (2,-2) -- (2,2);
\draw (4,-2) -- (4,2);
\draw (6,-2) -- (6,2);
\draw (8,-2) -- (8,2);
\end{scope}
\begin{scope}[yshift=-6.6cm]
\draw[>-<]  (-2,0) -- (10,0);
\draw (0,-2) -- (0,2);
\draw (2,-2) -- (2,2);
\draw (4,-2) -- (4,2);
\draw (6,-2) -- (6,2);
\draw (8,-2) -- (8,2);
\end{scope}
\begin{scope}[yshift=-8.8cm]
\draw[>-<]  (-2,0) -- (10,0) ;
\draw[<-] (0,-2) -- (0,2);
\draw[<-] (2,-2) -- (2,2);
\draw[<-] (4,-2) -- (4,2); 
\draw[<-] (6,-2) -- (6,2); 
\draw[<-] (8,-2) -- (8,2); 
\end{scope}
\end{tikzpicture}
\caption{Domain wall boundary condition.}
\label{fig:latticedw}
\end{figure}
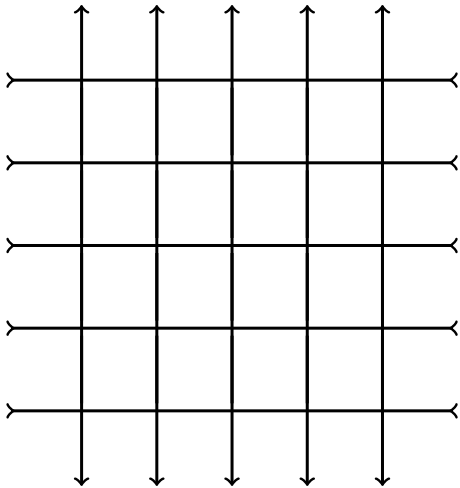

Using the prescriptions given in \secref{sec:vertexmodel}, we can see that the top and bottom boundary represent the states $\ket{\Psi_\uparrow}$ and $\ket{\Psi_\downarrow}$ defined in \eqref{states}. The left and right boundaries in Figure \ref{fig:latticedw} select the elements of the monodromy matrix forming the operator $\mathcal{O} = \prod_{i=1}^{L} B (\lambda_i , \{ \mu_j \})$. Thus our partition function can be written as
\be
\label{pf1}
Z(\{\lambda_i\}, \{\mu_j\}) = \bra{\Psi_{\downarrow}} \prod_{i=1}^{L} B (\lambda_i , \{ \mu_j \})  \ket{\Psi_{\uparrow}} \; ,
\ee
where we emphasise that the rapidities $\{\lambda_i\}$ may be chosen arbitrarily and do not have to satisfy the Bethe ansatz equations.
It is well known that \eqref{pf1} can be written as the Izergin-Korepin determinant \cite{Izergin87, Korepin82} as well as sums
over the permutation group \cite{Khoro_Paku05, Galleas11}.
 
\subsection{Scalar product}

The scalar product \eqref{eq:scprod} can be written in a similar fashion as the domain boundary partition function. From the definition of the Bethe eigenstates in \eqref{eq:BAeigvec} it follows immediately that the scalar product \eqref{eq:scprod} can be written as
\be
\label{scprod1}
N(\{\nu_i\},\{\lambda_j\})= \bra{\Psi_{\uparrow}} \prod_{i=1}^{n} C (\nu_i , \{ \mu_k \}) \prod_{i=1}^{n} B (\lambda_j , \{ \mu_k \})  \ket{\Psi_{\uparrow}} \; ,
\ee
where the sets of numbers $\{\nu_i\}$ and $\{\lambda_j\}$ each satisfy the Bethe equations \eqref{eq:bae}. Graphically \eqref{scprod1} is depicted in \figref{fig:latticescprod}.

\begin{figure}[h]\centering
\begin{tikzpicture}[scale=0.5,line width=0.4mm]
\draw[>-<] (-2,0) -- (12,0) node[right]{$\lambda_n$};
\draw[->] (0,-2) -- (0,2) node[above]{$\mu_1$};
\draw[->] (2,-2) -- (2,2);
\draw (4,-2) -- (4,2) ;
\draw (5,2) node[above]{\dots};
\draw (6,-2) -- (6,2) ;
\draw[->] (8,-2) -- (8,2);
\draw[->] (10,-2) -- (10,2) node[above]{$\mu_L$};
\begin{scope}[yshift=-2.2cm]
\draw (-2,0) node[left]{$\vdots$} -- (12,0) node[right]{$\vdots$};
\draw (0,-2) -- (0,2);
\draw (2,-2) -- (2,2);
\draw (4,-2) -- (4,2);
\draw (6,-2) -- (6,2);
\draw (8,-2) -- (8,2);
\draw (10,-2) -- (10,2);
\end{scope}
\begin{scope}[yshift=-4.4cm]
\draw[>-<] (-2,0) -- (12,0) node[right]{$\lambda_1$};
\draw (0,-2) -- (0,2);
\draw (2,-2) -- (2,2);
\draw (4,-2) -- (4,2);
\draw (6,-2) -- (6,2);
\draw (8,-2) -- (8,2);
\draw (10,-2) -- (10,2);
\end{scope}
\begin{scope}[yshift=-6.6cm]
\draw[<->] (-2,0) -- (12,0) node[right]{$\nu_1$};
\draw (0,-2) -- (0,2);
\draw (2,-2) -- (2,2);
\draw (4,-2) -- (4,2);
\draw (6,-2) -- (6,2);
\draw (8,-2) -- (8,2);
\draw (10,-2) -- (10,2);
\end{scope}
\begin{scope}[yshift=-8.8cm]
\draw (-2,0) node[left]{$\vdots$} -- (12,0) node[right]{$\vdots$};
\draw (0,-2) -- (0,2);
\draw (2,-2) -- (2,2);
\draw (4,-2) -- (4,2);
\draw (6,-2) -- (6,2);
\draw (8,-2) -- (8,2);
\draw (10,-2) -- (10,2);
\end{scope}
\begin{scope}[yshift=-11cm]
\draw[<->] (-2,0) -- (12,0) node[right]{$\nu_n$} ;
\draw[>-] (0,-2) -- (0,2);
\draw[>-] (2,-2) -- (2,2);
\draw (4,-2) -- (4,2); 
\draw (5,-2) node[below]{\dots};
\draw (6,-2) -- (6,2); 
\draw[>-] (8,-2) -- (8,2); 
\draw[>-] (10,-2) -- (10,2);
\end{scope}
\end{tikzpicture}
\caption{Scalar product boundary condition.}
\label{fig:latticescprod}
\end{figure}
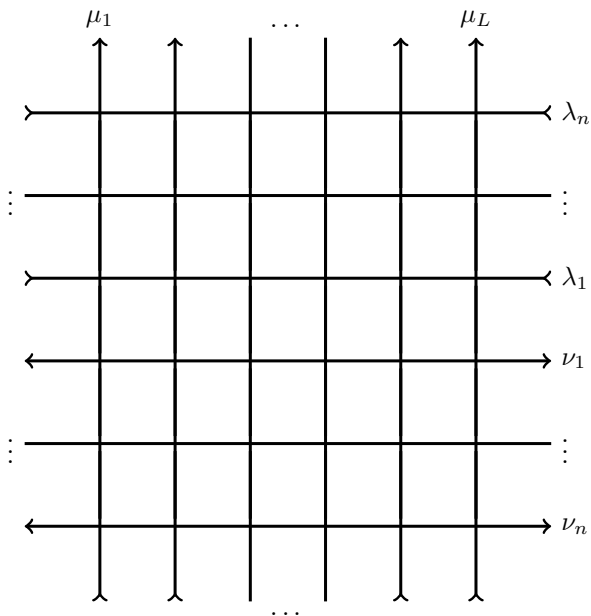
\subsection{Perimeter Bethe ansatz}
\label{sec:perimeter}

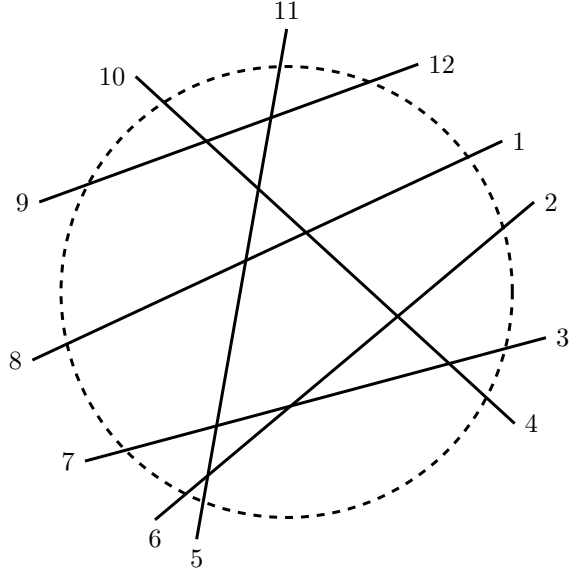
\begin{figure}[h]\centering
\begin{tikzpicture}[scale=0.5,line width=0.4mm]
\draw[dashed] (0,0) circle(6);
\draw (160:7) node[left]{9} -- (60:7) node[right]{12};
\draw (15:-7) node[left]{8} -- (35:7) node[right]{1};
\draw (40:-7) node[left]{7} -- (-10:7) node[right]{3};
\draw (60:-7) node[below]{6} -- (20:7) node[right]{2};
\draw (-30:7) node[right]{4} -- (125:7) node[left]{10};
\draw (90:7) node[above]{11} -- (-110:7) node[below]{5};
\end{tikzpicture}
\caption{Example of a general planar graph inside a domain $\mathcal{D}$ without multiple crossings at one point for $m=6$.}
\label{fig:Zinv}
\end{figure}

In this section, we review the result of Baxter \cite{Baxter87} for a general $Z$-invariant six-vertex model, adjusting some of his notation to serve our purposes \footnote{in particular, note that Baxter labels the endpoints in an anti-clockwise direction, and writes his formula with respect to the {\it left} rapidities}.  The general model is considered on a simply connected convex planar domain $\mathcal{D}$, with $m$ straight lines within it, starting and ending at the boundary \cite{Baxter78,Baxter87}. No three lines are allowed to intersect at a common point, (see Figure~\ref{fig:Zinv} for an example).  We fix some point on the boundary, and label the ends of the lines as $1,2,\dots,2m$, in a clockwise direction. The line with endpoints $i$ and $j$ is referred to as `the line $(i,j)$', where $1\leq i<j\leq2m$. We associate a `rapidity' $v_i$ with each endpoint $i$, and for each line $(i,j)$ impose the constraint
\be
v_j=v_i-\gamma.
\ee
Baxter uses the terminology `right' rapidity and `left' rapidity of a line $(i,j)$ for $v_i$ and $v_j$ respectively.  We also define the set of all rapidities $V=\{v_1,\dots,v_{2m} \}$, as well as the set $U$ of all right rapidities. 

A six-vertex model is constructed by placing arrows on the edges of the lattice so that each vertex (intersection) has two in-pointing arrows in and two out-pointing arrows (the ice-rule). A vertex at the crossing of lines $(i,j)$ and $(k,l)$ with $i<k<j<l$, as in Figure~\ref{fig:vertexinZ}, is given a weight $w(i,j|k,l)$ defined below. 
\begin{figure}[h]\centering
\begin{tikzpicture}[scale=0.2,line width=0.4mm]
\draw[dashed] (0,0) circle(6);
\draw[color=lightgray] (160:7) -- (60:7);
\draw[color=lightgray] (40:-7) -- (-10:7);
\draw[color=lightgray] (60:-7)  -- (20:7);
\draw[color=lightgray] (-30:7)  -- (125:7);
\draw (15:-7) node[left]{$j$} -- (35:7) node[right]{$i$};
\draw (90:7) node[above]{$l$}  -- (-110:7) node[below]{$k$} ;
\end{tikzpicture}
\caption{A vertex at the crossing of lines $(i,j)$ and $(k,l)$ with $i<k<j<l$ has weight $w(i,j|k,l)$.}
\label{fig:vertexinZ}
\end{figure}
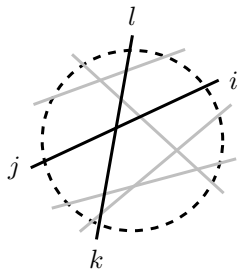
The six possible arrangements as viewed from the boundary points $k$ and $j$ are shown in Figure~\ref{fig:6vertexZinv}. 
\begin{figure}[h]\centering
\begin{tikzpicture}[scale=0.5,line width=0.4mm]
\draw[>->] (15:-1) -- (35:1);
\draw[<-<] (90:1) -- (-110:1) node[below]{$w_1$};
\draw[xshift=4cm,<-<] (15:-1) -- (35:1);
\draw[xshift=4cm,>->] (90:1) -- (-110:1) node[below]{$w_2$};
\draw[xshift=8cm,>->] (15:-1) -- (35:1);
\draw[xshift=8cm,>->] (90:1) -- (-110:1) node[below]{$w_3$};
\draw[xshift=12cm,<-<] (15:-1) -- (35:1);
\draw[xshift=12cm,<-<] (90:1) -- (-110:1) node[below]{$w_4$};
\draw[xshift=16cm,>-<] (15:-1) -- (35:1);
\draw[xshift=16cm,<->] (90:1) -- (-110:1) node[below]{$w_5$};
\draw[xshift=20cm,<->] (15:-1) -- (35:1);
\draw[xshift=20cm,>-<] (90:1) -- (-110:1) node[below]{$w_6$};
\end{tikzpicture}
\caption{Six vertex configurations on vertices inside $\mathcal{D}$.}
\label{fig:6vertexZinv}
\end{figure}
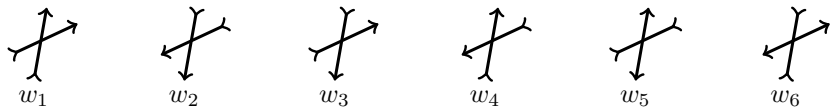
Boltzmann weights are assigned to the six configurations as follows (in the notation of this paper):
\begin{align}
w_1=w_2 &=1,\nn \\
w_3=w_4 &=\frac{b(v_l-v_k)}{a(v_l-v_k)} = \frac{\sinh(v_l-v_k)}{\sinh(v_l-v_k+\gamma)},\\
w_5=w_6 &=\frac{c(v_l-v_k)}{a(v_l-v_k)}=\frac{\sinh(\gamma)}{\sinh(v_l-v_k+\gamma)}.\nn
\end{align}

We also fix boundary conditions, with $m$ arrows pointing into $\mathcal{D}$ and $m$ arrows pointing out.  We specify these boundary arrows with the set of endpoint locations where an out arrow occurs, $X=\{ x_1,\dots,x_m \}$.  The partition function of the model is then a function of $X$ and $V$.  We now come to the main result (for our purposes) of \cite{Baxter87}:
\begin{theorem}[Baxter]
\label{th:baxter}
Let $V=\{v_1,\ldots,v_{2m}\}$ be the set of all rapidities, and let $U$ be the set of right rapidities. The partition function $Z(X|V)$ for the Z invariant six-vertex model is given by
\begin{align}
Z(X|V)&:= \sum_{\text{cfgs.}} \prod_{(i,j|k,l)} w(i,j|k,l) \nn \\
&=   C^{-1} \sum_{P} \prod_{1\leq i<j\leq m} \frac{a(u_j - u_i)}{b(u_j - u_i)}   \prod_{i=1}^m \phi_{x_i}(u_i) 
\end{align}
where
\begin{align}
\label{eq:phidef}
\phi_x(u)&= \prod_{j=1}^{x-1} a(v_j-u) \prod_{j=x+1}^{2m} b(v_j-u),
\end{align}
and the pre-factor is defined by
\be
C=(-1)^{n_\uparrow}\sinh(\gamma)^{-m}\prod_{1\leq i,j\leq m} a(u_i-u_j)^2,
\ee
where $n_\uparrow$ is the number of out-point arrows with coordinate $x\leq m$ and the variables $u_i$ refer to the elements of the set of right rapidities $U$.
\end{theorem}

\section{Integral expression}
In this section we will derive a multiple integral expression for the partition function \eqref{partition2} in the case of fixed boundary conditions. Theorem~\ref{th:baxter} provides an explicit expression derived by Baxter \cite{Baxter87} for such partition functions for general graphs and boundary conditions. This expression is in terms of a large sum over the symmetric group and in this form not immediately useful. We hope our integral expressions are more manageable and may lead to further progress for correlation functions as well as higher rank cases. 

We wish to write the sum over permutations as an integral, summing over residues.  We introduce auxiliary integration variables $w_i$, $i=1,\dots,m$. This leads to the following result.
\begin{theorem}
Let $V=\{v_1,\ldots,v_{2m}\}$ be the set of all rapidities, and let $U$ be the set of left rapidities.
\begin{multline}
Z(X|V) = \frac{C^{-1} }{(2\pi\ii)^m}\oint \dots \oint \frac{\prod_{i=1}^m\prod_{j=i+1}^{m}  a(w_j-w_i)b(w_i-w_j)} {\prod_{i,j=1}^{m} b(w_i-u_j) } \prod_{i=1}^{m} \phi_{x_i}(w_i)\  \dd w_1 \dots \dd w_{m}.
\label{eq:Baxter_int}
\end{multline}
The contours of integration are around the poles at $w_i=u_j$, $i,j=1,\ldots,m$ and the prefactor is again given by
\be
C=(-1)^{n_\uparrow}\sinh(\gamma)^{-m}\prod_{1\leq i,j\leq m} a(u_i-u_j)^2,
\ee
and $\phi_x(w)$ and $n_\uparrow$ are as in Theorem~\ref{th:baxter}.
\end{theorem}
\begin{proof}
The factor $\prod_{j>i} b(w_i-w_j)$ in the denominator implies that the only nonzero contributions arise from residues at different poles, i.e. if $\{w_1,\ldots,w_m\}$ is a permutation of $\{u_1,\ldots,u_m\}$. 
\end{proof}

\subsection{Domain wall boundary conditions}
\label{sec:integralformDWBC}

We now consider the result of Baxter \cite{Baxter87}, specialised to a $L$ by $L$ square lattice with domain wall boundary conditions. Using the notation $\bar{\mu}_j=\mu_{L-j+1}$, we have $4L$ rapidities though only $2L$ are independent. The full set of rapidities are given by
\be
V=\{\mu_1,\dots,\mu_L, \bar{\lambda}_1,\dots,\bar{\lambda}_L,\bar{\mu}_1-\gamma,\dots,\bar{\mu}_L-\gamma,\lambda_1-\gamma,\dots,\lambda_L-\gamma  \}
\ee
while the set of left rapidities is
\bea
U=\{ u_1,\dots,u_{2L} \}=\{\mu_1,\dots,\mu_L ,\bar{\lambda}_1,\dots,\bar{\lambda}_L\}.
\eea
We take boundary conditions of out arrows being on the vertical lines, and in arrows on the horizontal lines
according to \figref{fig:latticedw}.  Thus the positions of the 
out arrows are
\bea
X=\{ 1,\dots, L, 2L+1, \dots, 3L  \}.
\eea
With these definitions we have that the domain wall boundary partition function \eqref{pf1} is equal to
\be
Z(\{\lambda_i\},\{\mu_j\}) = \prod_{i,j=1}^L a(\lambda_i-\mu_j) Z(X|V).
\ee

The domain wall boundaries specialisation of \eqref{eq:Baxter_int} implies that the only nonzero contributions arise from residues at 
$w_i=\mu_i$ for $i=1,\ldots,L$, and if $\{w_{L+1},\ldots,w_{2L}\}$ is a permutation of $\{\lambda_{1},\ldots,\lambda_{L}\}$.  With this 
specialisation, and the re-labelling of the integration variables $w_{L+i} \rightarrow w_i$, the formula becomes
\begin{multline}\label{eq:int}
Z(\{\lambda_i\},\{\mu_j\}) = \left(\frac{\sinh\gamma}{2\pi \ii}\right)^L  \oint \dots \oint \frac{\prod_{i=1}^L\prod_{j=i+1}^{L} a({w}_j-{w}_i)b({w}_j-{w}_i)} {\prod_{i,j=1}^L b({w}_i-\lambda_j)}\\
\times \prod_{i=1}^L \prod_{j=1}^{i-1}  b(\bar{\mu}_j-{w}_i)\prod_{j=i+1}^{L}  a({w}_i-\bar{\mu}_j)\  \dd{w}_1 \dots \dd {w}_L
\end{multline}
where the integrals are around the poles at ${w}_i=\lambda_j$.

\subsubsection{Homogeneous limit}
\label{sec:homogeneous}
The integral formula \eqref{eq:int}, enables us to readily calculate the homogeneous limit of the partition function. With the changes of variables
\bea\label{eq:cofv}
w=\frac{1}{2}\log \bigg( \frac{1-t x}{t- x} \bigg),\quad \lambda = \frac{1}{2}\log \bigg(\frac{1-t z}{t- z}\bigg),\quad \mu = \frac{1}{2}\log \bigg(\frac{1-t y}{t- y}\bigg),\quad t=\e^{\gamma},
\eea
we obtain
\begin{multline}
Z= \frac{c^{L^2}}{(2\pi \ii)^L} \oint \dots \oint \frac{\prod_{i=1}^L\prod_{j=i+1}^{L} (x_i-x_j)(1+\kappa x_j+x_ix_j)} {\prod_{i,j=1}^L (z_j-x_i)} \\
\times  \prod_{i=1}^L\prod_{j=1}^L \frac{(y_j-z_i)(1+\kappa z_i+ y_jz_i)}{(1+\kappa y_j +y_j^2)^{1/2}(1+\kappa z_i +z_i^2)^{1/2}} \\
\times \prod_{i=1}^L \bigg( \prod_{j=1}^i \frac{1}{(1+\kappa x_i +x_i y_j)}  \prod_{j=i}^L \frac{1}{(y_j-x_i)} \bigg)\ \dd x_1 \dots \dd x_L,
\end{multline}
where $\kappa=-(t+1/t)$.  

We now set $z_i=z$ and $y_i=y$ for all $i$ ({\it i.e. } $\gl_i=\gl$, $\mu_i=\mu$),
 \begin{multline}
Z= \frac{c^{L^2}}{(2\pi \ii)^L} \oint \dots \oint \frac{\prod_{i=1}^L \prod_{j=i+1}^{L} (x_i-x_j)(1+\kappa x_j+x_ix_j)} {\prod_{i=1}^L (z-x_i)^L} \\
 \times \frac{(y-z)^{L^2}(1+\kappa z+ y z)^{L^2}}{(1+\kappa y +y^2)^{L/2}(1+\kappa z +z^2)^{L/2}}\prod_{i=1}^L  \frac{1}{(z_j-x_i)} \\
\times \prod_{j=1}^i \frac{1}{(1+\kappa x_i +x_i y)^L}  \prod_{j=i}^L \frac{1}{(y-x_i)^{L+1-i}}\ \dd x_1 \dots \dd x_L.
\end{multline}
Also the limit $a=b=c=1$, where the partition function counts the number of alternating sign matrices \cite{Kuperberg96}, is obtained as follows.  We first set $a(\lambda,\mu)=b(\lambda,\mu)$, {\it i.e.} $\lambda-\mu=-\gamma/2+\ii \pi/2$.  In the transformed variables we achieve this by setting $z=0$ and $y=1$, corresponding to $\lambda=-\gamma/2$ and $\mu=-\ii\pi/2$.  We also normalise each of these weights to $1$ by dividing through by $a(\lambda,\mu)$ and thus obtain
\begin{multline}\label{eq:hom}
Z_{\rm ASM}= \frac{(c/a)^{L^2}}{(2+\kappa)^{L^2/2}(2\pi \ii)^L} \oint \dots \oint \frac{\prod_{i=1}^L\prod_{j=i+1}^{L} (x_i-x_j)(1+\kappa x_j+x_ix_j)} {\prod_{i=1}^L (-x_i)^{L}} \\
\times  \prod_{i=1}^L (1+(\kappa +1)x_i )^i (1-x_i)^{-L+i-1} \ \dd x_1 \dots \dd x_L.
\end{multline}
(Note that we would finally get $(a=b=c=1)$ by setting $\gamma=\ii \pi/3$, {\it i.e.} $\kappa=-1$). Alternatively we set $\kappa=\tau-2$, and write equation \eqref{eq:hom} as the following constant term expression:
\begin{multline}
A_L(\tau)= {\rm CT}\left[ \prod_{i=1}^L\prod_{j=i+1}^{L} (x_j-x_i)(1+(\tau-2) x_j+x_ix_j)\right.\\ 
\left.\prod_{i=1}^L \frac{1}{x_i^{L-1}(1-x_i)^{L-i+1}(1+(\tau-1) x_i)^{i}} \right],
\end{multline}
which yields polynomials in $\tau$ corresponding to the enumeration of alternating sign matrices, by number of $-1$s. This latter formula is very similar to expressions obtained in \cite{DiFZJ2007,Zeilb07,FonsecaZJ09,deGier_Lascoux_Sorrell10,FonsecaNadeau11} as polynomial solutions of the $q$-deformed Knizhnik-Zamolodchikov equation and their relation to the combinatorics of alternating sign matrices and symmetric plane partitions.
\subsection{Scalar Product}
\label{sec:scalar_product}
Next we consider the result of Baxter \cite{Baxter87} specialised to the $2n \times L$ rectangular lattice depicted in  
\figref{fig:latticescprod}. This specialisation renders the norm of Bethe vectors (\ref{scprod1}), and according 
to the conventions of \Secref{sec:perimeter} we set $m=L+2n$. We have rapidities
\bea
\label{vv}
V=\{ \mu_1,\dots, \mu_{L}, \bar{\lambda}_1 , \dots , \bar{\lambda}_n , \nu_1 , \dots , \nu_n, \bar{\mu}_1 - \gamma , \dots , \bar{\mu}_L - \gamma, \bar{\nu}_1 -\gamma, \dots , \bar{\nu}_n -\gamma , \lambda_1 -\gamma , \dots , \lambda_n -\gamma \} \nn \\
\eea
and the set of left rapidities is then given by
\bea
\label{uu}
U=\{ \mu_1,\dots, \mu_{L}, \bar{\lambda}_1 , \dots , \bar{\lambda}_n , \nu_1 , \dots , \nu_n \} \; .
\eea
Here we use the definitions $\bar{\mu}_j = \mu_{L+1-j}$, $\bar{\lambda}_j = \lambda_{n+1-j}$ and $\bar{\nu}_j = \nu_{n+1-j}$. In order to characterise
the scalar product (\ref{scprod1}) we also need to specify the location of the out arrows $X$. 
We have these out arrows appearing in three distinct blocks with positions:
\bea\label{xx}
X=\{ 1,\dots,L, L+n+1, \dots, L+2n,2L+2n+1,\dots,2L+3n \}
\eea

Similarly to the case of domain wall boundaries, the norm (\ref{scprod1}) consists of the partition function $Z(X|V)$
up to a normalisation factor arising from the normalisation of the statistical weights. Thus we have
\be
N(\{\nu_i\},\{\lambda_j\}) = \prod_{i=1}^n \prod_{j=1}^L a(\lambda_i-\mu_j) a(\nu_i-\mu_j)  Z(X|V).
\ee
with $Z(X|V)$ computed using (\ref{vv}), (\ref{uu}) and (\ref{xx}).  With this specialisation the integrand of (\ref{eq:Baxter_int})
has poles only at $w_i = \mu_i$ for $i=1, \dots , L$, (see Appendix \ref{spc}) and thus the integration over this particular subset of variables
is trivial. After performing the trivial integrations and some re-arrangments detailed in Appendix~\ref{spc} we are left with

\begin{align}
\label{eq:intsp}
N(\{\nu_i\},\{\lambda_j\}) = & \frac{(-1)^{(L-1)n} (\sinh{\gamma})^{2n} }{(2\pi \ii)^{2n}} \prod_{i,j=1}^{n} \frac{1}{a(\lambda_i - \nu_j) a(\nu_j - \lambda_i) a(\nu_i - \nu_j)^2}   \nn \\
& \times \oint \dots \oint \dd{w}_{1} \dots \dd w_{2n}  \frac{\prod_{i=1}^{2n}\prod_{j=i+1}^{2n} a({w}_{j}-{w}_{i}) b({w}_{j}-{w}_{i})} {\prod_{i=1}^{2n} \prod_{j=1}^{n}  b({w}_{i}-\lambda_j) b({w}_{i}-\nu_j) } \nn\\
& \quad \times \prod_{i=1}^{n} \prod_{j=1}^{L} a({w}_{i}-\mu_j) b(\mu_j - {w}_{n+i}) \prod_{i=1}^{n} \prod_{j=1}^{n} a({w}_{i}-\nu_j)a(\nu_j - {w}_{n+i})\nn\\
&\quad\times \prod_{i=1}^{n} \prod_{j=i+1}^{n} b(\nu_j - {w}_{i}) b({w}_{2n+1-i} - \nu_j ) a(\nu_i - {w}_{j}) a({w}_{2n+1-j} - \nu_i )  \; .
\end{align}

\section{Three-point functions in $\mathcal{N}=4$ Super Yang-Mills} 
\label{sec:3pt}

The formul\ae (\ref{eq:intsp}) for the scalar product of Bethe vectors is valid for arbitrary values of the complex
parameters $\{\nu_i\}$ and $\{\lambda_j\}$. On the other hand, the expression (\ref{eq:intsp}) will yield 
the norm of the transfer matrix eigenstate (\ref{eq:BAeigvec}) only when the set $\{\nu_i\}$ is related 
to the set $\{\lambda_j\}$, i.e. $\nu_j = \lambda^{*}_j$, and with the set $\{\lambda_j\}$ subjected
to the Bethe ansatz equations (\ref{eq:bae}). In that case Slavnov's formula \cite{Slavnov_89}  provides a rather simple expression for the scalar product in terms of a determinant of a $n \times n$ matrix. For generic parameters $\{\nu_i\}$ and $\{\lambda_j\}$, one would have to consider the formul\ae of \cite{Korepin82,Korepin_book}, given in terms of a large sum over partitions. However, the recent developments in the computation 
of three-point functions in $\mathcal{N}=4$ Super Yang-Mills (SYM)  seems to require manageable expressions for the norms of Bethe vectors for general values of the parameters and here we shall illustrate how the integral formula (\ref{eq:intsp}) can be embedded in the framework of \cite{Escobedo_2011}.

Three-point correlation functions of primary operators in $\mathcal{N}=4$ SYM theory are constrained by conformal invariance to have the form
\be
\langle \mathcal{O}_i (x_i) \mathcal{O}_j (x_j) \mathcal{O}_k (x_k) \rangle \; = \frac{\sqrt{\mathcal{N}_i \mathcal{N}_j \mathcal{N}_k} \; c_{ijk}}{| x_i - x_j |^{\Delta_i + \Delta_j - \Delta_k } | x_j - x_k |^{\Delta_j + \Delta_k - \Delta_i } | x_k - x_i |^{\Delta_k + \Delta_i - \Delta_j }},
\ee
where $\mathcal{N}_j$ are normalization factors and $\Delta_i$ are the respective conformal dimensions. The dimensions  $\Delta_i$ are already fixed in the computation of two-point functions and a large literature is devoted to that problem in the planar limit of the $\mathcal{N}=4$ SYM, exploring integrable structures. For a detailed revision on this subject see for instance \cite{Beisert} and references therein. Here the quantity we wish to 
evaluate is the structure constant $c_{ijk}$, which has been worked out to leading order in \cite{Escobedo_2011} employing the spin chain picture of single trace operators introduced in \cite{Minahan_2002}. See also \cite{Okuyama_2004} and \cite{Roiban_2004} for previous works on $\mathcal{N}=4$ SYM three-point functions within the spin chain picture. 

The structure constant $c_{ijk}$ in the $\mathcal{N}=4$ SYM admits the perturbative expansion 
$c_{ijk} = \frac{1}{N_c} c^{(0)}_{ijk} + \frac{\theta}{N_c} c^{(1)}_{ijk} + \frac{\theta^2}{N_c} c^{(2)}_{ijk} + \ldots$,
where $N_c$ and $\theta$ are respectively the number of colours and the 't Hooft coupling constant. 
We shall consider the leading order term $c^{(0)}_{123}$ as described in \cite{Escobedo_2011} and in order to
compute it from our results we first need to consider a special limit usually refereed as rational limit. We start by
performing the change of variables
\be
\label{rat1}
\lambda = \bar{\gamma} (u - \tfrac{\ii}{2}) \quad \quad \mbox{and} \quad \quad  \gamma = \ii \bar{\gamma} \; ,
\ee
and we also set the inhomogeneities $\mu_j = 0$.

From (\ref{su2})-(\ref{monorep}) we find that
\be
\label{rat2}
\lim_{\genfrac{}{}{0pt}{}{\bar{\gamma} \rightarrow 0}{\mu_j \rightarrow 0}}   \frac{1}{\sinh{\bar{\gamma}}^{L}} \mathcal{T}(\lambda) = \overline{\mathcal{T}}(u) \; ,
\ee
where $\overline{\mathcal{T}}(u)$ corresponds to the monodromy matrix considered in \cite{Escobedo_2011}.
In practice the above limit corresponds to considering our results with statistical weights $a$, $b$ and $c$ given by 
\begin{eqnarray}
\label{rat}
a &=& u + \tfrac{\ii}{2}, \nn \\
b &=& u - \tfrac{\ii}{2}, \nn \\
c &=& \ii \; .
\end{eqnarray}

Now we consider three independent $XXX$-type spin chains with periodic boundary conditions which we shall
refer as $\mathcal{A}_1$, $\mathcal{A}_2$ and $\mathcal{A}_3$. Those spin chains can be obtained as the 
logarithmic derivative of the transfer matrix (\ref{transfmatrix}) up to normalisation factors after 
considering the rational limit above described. This is a standard construction and a detailed description
can be found in  \cite{Korepin_book}.
Within the spin chain picture introduced in \cite{Minahan_2002}, each single trace operator $\mathcal{O}_i$ in 
the $SU(2)$ sector of the $\mathcal{N}=4$ SYM at one-loop order is then represented by the spin chain $\mathcal{A}_i$
under the zero momentum condition.

The state of each spin chain $\mathcal{A}_j$ with $L_j$ sites is characterised by rapidities satisfying Bethe
ansatz equations. For instance we use the set of $m_1$ rapidities $\{ u_1 , \dots , u_{m_1} \}$ to characterise 
the spin chain $\mathcal{A}_1$, the $m_2$ rapidities $\{ v_1 , \dots , v_{m_2} \}$ for $\mathcal{A}_2$ and $m_3$ 
rapidities $\{ w_1 , \dots , w_{m_3} \}$ for $\mathcal{A}_3$. They are subjected to the following Bethe ansatz equations:
\begin{align}
\label{eq:BAE}
\left[ \frac{a(u_i)}{b(u_i)} \right]^{L_1} &= \prod_{\stackrel{k=1}{k \neq i}}^{m_1} - \frac{a(u_i - u_k + \frac{\ii}{2})}{a(u_k - u_i + \frac{\ii}{2}) },
\qquad \qquad i = 1, \dots , m_1\; , \nn \\ 
\left[ \frac{a(v_i)}{b(v_i)} \right]^{L_2} &= \prod_{\stackrel{k=1}{k \neq i}}^{m_2} - \frac{a(v_i - v_k + \frac{\ii}{2})}{a(v_k - v_i + \frac{\ii}{2}) },
\qquad \qquad i = 1, \dots , m_2\; , \nn \\ 
\left[ \frac{a(w_i)}{b(w_i)} \right]^{L_3} &= \prod_{\stackrel{k=1}{k \neq i}}^{m_3} - \frac{a(w_i - w_k + \frac{\ii}{2})}{a(w_k - w_i + \frac{\ii}{2}) },
\qquad \qquad i = 1, \dots , m_3 \; ,
\end{align}
and the zero momentum condition translates into
\be
\prod_{i=1}^{m_1} \frac{a(u_i)}{b(u_i)} = \prod_{i=1}^{m_2} \frac{a(v_i)}{b(v_i)} = \prod_{i=1}^{m_3} \frac{a(w_i)}{b(w_i)} = 1  \; .
\ee
We remark that now we are considering the statistical weights $a$, $b$ and $c$ given in (\ref{rat}).
Next we define the functions
\begin{align}
E_2 &= \prod_{j=1}^{m_2} \left[ \frac{a(v_j)}{b(v_j)} \right]^{L_2},   & F_2 &= \prod_{i=1}^{m_2} \prod_{j=i+1}^{m_2}  \frac{a(v_i - v_j + \frac{\ii}{2})}{b(v_i - v_j + \frac{\ii}{2})}, \nn \\
F_1 &= \prod_{i=1}^{m_1} \prod_{j=i+1}^{m_1}  \frac{a(u_i - u_j + \frac{\ii}{2})}{b(u_i - u_j + \frac{\ii}{2})}, & G_2 &= \prod_{j=1}^{m_2} \prod_{i=j+1}^{m_2}  \frac{a(v_i - v_j + \frac{\ii}{2})}{b(v_i - v_j + \frac{\ii}{2})},
\end{align}
and separate the set of rapidities $\{ u \} = \{ u_1 , \dots , u_{m_1} \}$ into two subsets
$\{ \alpha \} $ and $\{ \bar{\alpha} \} $ such that $\{ \alpha \} \cup \{ \bar{\alpha} \} = \{ u \}$. There are $2^{m_1}$
ways to partitionate $\{ u \}$ into $\{ \alpha \} $ and $\{ \bar{\alpha} \}$, and each possibility is simply recasted as
\be
\{ \alpha \} = \{ u_1 ,  \dots , u_{n_1} \} \qquad \mbox{and} \qquad \{ \bar{\alpha} \} = \{ \bar{u}_1 ,  \dots , \bar{u}_{\bar{n}_1} \} ,
\ee
with $m_1 = n_1 + \bar{n}_1$. For each partitioning we also define the functions
\begin{equation}
E(\bar{\alpha}) = \prod_{j=1}^{\bar{n}_1} \left[ \frac{a(\bar{u}_j)}{b(\bar{u}_j)} \right]^{L_1 +1},  \qquad \qquad F(\alpha ,\bar{\alpha}) = \prod_{i=1}^{n_1} \prod_{j=1}^{\bar{n}_1} \frac{a(u_i - \bar{u}_j + \frac{\ii}{2})}{b(u_i - \bar{u}_j + \frac{\ii}{2})} \; \; ,
\end{equation}
and following \cite{Escobedo_2011} the structure constant for three-point functions at leading order in planar $\mathcal{N}=4$ SYM is given by
\be
\label{eq:c0}
c^{(0)}_{123} = \sqrt{\frac{L_1 L_2 L_3}{\mathcal{N}_1 \mathcal{N}_2 \mathcal{N}_3}} \frac{E_2 G_2}{F_2 F_1} (-1)^{m_1} \sum_{\alpha ,  \bar{\alpha}} 
E(\bar{\alpha}) F(\alpha , \bar{\alpha})  S_2 (\alpha) S_3 (\bar{\alpha}) ,
\ee
under the constraint $m_1 =m_2 + m_3$. The normalisation factors are
\begin{align}
\label{eq:norms}
\mathcal{N}_1 = & (-1)^{m_1} \prod_{j=1}^{m_1}  b(u_j)^{-L_1} a(u^{*}_j)^{-L_1} a(u_j) b(u^{*}_j)  \prod_{i=1}^{m_1} \prod_{j=i+1}^{m_1} \frac{b(u_i - u_j + \frac{\ii}{2})}{a(u_i - u_j + \frac{\ii}{2})} \nn \\
& \times \prod_{j=1}^{m_1} \prod_{i=j+1}^{m_1} \frac{b(u^{*}_i - u^{*}_j + \frac{\ii}{2})}{a(u^{*}_i - u^{*}_j + \frac{\ii}{2})}  N_{L_1}(\{ u^{*} \} , \{ u \} ), \nn \\
\mathcal{N}_2 = & (-1)^{m_2} \prod_{j=1}^{m_2}  b(v_j)^{-L_2} a(v^{*}_j)^{-L_2} a(v_j) b(v^{*}_j) \prod_{i=1}^{m_2} \prod_{j=i+1}^{m_2} \frac{b(v_i - v_j + \frac{\ii}{2})}{a(v_i - v_j + \frac{\ii}{2})}  \nn \\
& \times  \prod_{j=1}^{m_2} \prod_{i=j+1}^{m_2}  \frac{b(v^{*}_i - v^{*}_j + \frac{\ii}{2})}{a(v^{*}_i - v^{*}_j + \frac{\ii}{2})} N_{L_2}(\{ v^{*} \} , \{ v \} ), \nn \\
\mathcal{N}_3 = &  (-1)^{m_3} \prod_{j=1}^{m_3} b(w_j)^{-L_3} a(w^{*}_j)^{-L_3} a(w_j) b(w^{*}_j) \prod_{i=1}^{m_3} \prod_{j=i+1}^{m_3} \frac{b(w_i - w_j + \frac{\ii}{2})}{a(w_i - w_j + \frac{\ii}{2})}  \nn \\
& \times \prod_{j=1}^{m_3} \prod_{i=j+1}^{m_3} \frac{b(w^{*}_i - w^{*}_j + \frac{\ii}{2})}{a(w^{*}_i - w^{*}_j + \frac{\ii}{2})} N_{L_3}(\{ w^{*} \} , \{ w \} )  \; \; ,
\end{align}
where $N_{L}$ denotes the scalar product of Bethe vectors (\ref{eq:intsp}) for a lattice of length $L$ after the limit (\ref{rat1}).
The function $N_{L}$ turns out to be explicitly given by
\begin{align}
\label{eq:splim}
N_L (\{\nu \},\{\lambda\})& =  \frac{1}{(2\pi \ii)^{2n}}  \prod_{i,j=1}^{n} \frac{1}{a(\lambda_i - \nu_j+ \frac{\ii}{2}) a(\nu_j - \lambda_i+ \frac{\ii}{2}) a(\nu_i - \nu_j+ \frac{\ii}{2})^2}   \nn \\
& \times \oint \dots \oint \dd{w}_{1} \dots \dd w_{2n}  \frac{\prod_{i=1}^{2n}\prod_{j=i+1}^{2n} a({w}_{j}-{w}_{i}+ \frac{\ii}{2}) b({w}_{j}-{w}_{i}+ \frac{\ii}{2})} {\prod_{i=1}^{2n} \prod_{j=1}^{n}  b({w}_{i}-\lambda_j+ \frac{\ii}{2}) b({w}_{i}-\nu_j+ \frac{\ii}{2}) } \nn\\
& \; \; \times \prod_{i=1}^{n} [ a({w}_{i}) b({w}_{n+i}) ]^L \prod_{i=1}^{n} \prod_{j=1}^{n} a({w}_{i}-\nu_j+ \tfrac{\ii}{2})a(\nu_j - {w}_{n+i}+ \tfrac{\ii}{2})\nn\\
& \; \; \times \prod_{i=1}^{n} \prod_{j=i+1}^{n} b(\nu_j - {w}_{i}+ \tfrac{\ii}{2}) b({w}_{2n+1-i} - \nu_j+ \tfrac{\ii}{2} ) a(\nu_i - {w}_{j}+ \tfrac{\ii}{2}) a({w}_{2n+1-j} - \nu_i + \tfrac{\ii}{2}),
\end{align}
with the functions $a$ and $b$ as described in (\ref{rat}). The number of excitations $n$ is understood by the number of rapidities in the argument.

In its turn the function $S_2 (\alpha)$ is given by
\begin{align}
\label{eq:s2}
S_2 (\alpha) = & \prod_{j=1}^{n_1} b(u_j)^{m_3 - L_1}  a(u_j) \prod_{j=1}^{m_2} a(v_j)^{m_3 - L_1} b(v_j)  
\prod_{j=1}^{m_2} \prod_{i=j+1}^{m_2} \frac{b(v_i - v_j + \frac{\ii}{2})}{a(v_i - v_j + \frac{\ii}{2})} \nn \\
& \times N_{L_1 - m_3} (\{ v \} , \{ u \}) \; \delta_{n_1 , m_2}\; ,
\end{align}
while
\begin{align}
\label{eq:s3}
S_3 (\bar{\alpha}) = &  \prod_{j=1}^{m_3} b(w_j)^{-m_3} a(w_j) \prod_{j=1}^{\bar{n}_1} a(\bar{u}_j)^{-m_3} b(\bar{u}_j)  \prod_{i=1}^{m_3} \prod_{j=i+1}^{m_3} \frac{b(w_i - w_j + \frac{\ii}{2})}{a(w_i - w_j + \frac{\ii}{2})} \nn \\
& \times Z_{m_3} (\{ \bar{u} \} ) Z_{m_3} (\{ w \} )   \; \delta_{\bar{n}_1 , m_3} \; .
\end{align}
In the expressions (\ref{eq:s2}) and (\ref{eq:s3}), $\delta_{i,j}$ stands for the standard Kronecker delta and $Z_{m_3}$ consists of the
six-vertex model partition function with domain wall boundaries (\ref{eq:int}) for a lattice with dimensions $m_3 \times m_3$ in the limit (\ref{rat1})-(\ref{rat}) previously discussed. 
For completeness it is given by
\begin{multline}
\label{eq:int1}
Z_{m_3} (\{\lambda\}) = \frac{(-1)^{m_3 (m_3 -1)/2}}{(2\pi)^{m_3}}  \oint \dots \oint \frac{\prod_{i=1}^{m_3}\prod_{j=i+1}^{m_3} a({w}_j-{w}_i + \tfrac{\ii}{2})b({w}_j-{w}_i+ \tfrac{\ii}{2})} {\prod_{i,j=1}^{m_3} b({w}_i-\lambda_j+ \frac{\ii}{2})} \\
\times\prod_{i=1}^{m_3} b({w}_i)^{i-1} a({w}_i)^{m_3 -i}\  \dd{w}_1 \dots \dd{w}_{m_3} \; ,
\end{multline}
recalling that $a$ and $b$ correspond to the functions in (\ref{rat}).

In this way, to compute the structure constant leading term $c^{(0)}_{123}$ according to the prescriptions of \cite{Escobedo_2011},
one needs to substitute (\ref{eq:splim}) in (\ref{eq:norms}) and (\ref{eq:s2}), and then insert the results into (\ref{eq:c0}). Similarly we also
need to substitute (\ref{eq:int1}) in (\ref{eq:s3}) and replace the results in  (\ref{eq:c0}). In summary, the evaluation of $c^{(0)}_{123}$ is
thus reduced to the evaluation of the countour integrals (\ref{eq:splim}) and (\ref{eq:int1}).

\section{Conclusion}

The main result of this paper is to write the partition function for a general $Z$-invariant six-vertex model as a multiple contour integral over a factorised polynomial kernel. This allows us to write a new representation for the partition function of the six vertex model with domain wall boundary conditions as a multiple contour integral. Likewise we derive an off-shell  multiple integral expression for the scalar product of two Bethe vectors of the six-vertex model transfer matrix. 

The study of two-point functions in $\mathcal{N}=4$ Super Yang-Mills (SYM) has advanced dramatically
due to the presence of integrable structures, and the computation of three-point functions at weak coupling seems to be benefited as well \cite{Escobedo_2011}.
In this sense we have also illustrated in Section~\ref{sec:3pt} how the integral formul\ae presented in 
Sections~\ref{sec:integralformDWBC} and \ref{sec:scalar_product} can be used in that context. Although in a different fashion, recently the connection was realised between the partition function of a system of statistical mechanics and the computation of norms of Bethe vectors in a particular limit \cite{Gromov_2011}. We hope the results presented here to help understanding this relation in more general cases.

The integral formul\ae we obtain are very similar to those in multi-variate polynomial solutions of the $q$-deformed Knizhnik-Zamolodchikov 
equation, a connection which deserves future exploration. In \cite{DiFZJ2007,Zeilb07,FonsecaZJ09,deGier_Lascoux_Sorrell10,FonsecaNadeau11} methods have been developed to compute with such multi-variate expressions which we hope will be applicable in the current context as well. We further hope that our representation for the scalar product can be extended to compute correlation functions, providing an alternative to the methods developed in \cite{KitanineMT99,IzerginKMT99,KitanineMT00}, and also constitutes an approach which may be generalised to higher rank solvable lattice models. 

\section*{Acknowledgments}
Our warm thanks goes to Michael Wheeler for instructive discussions, and we gratefully acknowledge financial support from the Australian Research Council (ARC). 

\appendix

\section{Details of scalar product computation}\label{spc}

We consider the result of Baxter \cite{Baxter87}, specialised to a $2n$ by $L$ square lattice.  {\it i.e.} we have $2n+L$ lines and 
$4n+2L$ rapidities.  (see \figref{fig:latticescprod} ) \\

Thus we have the set of right rapidities (in a clockwise direction), using the notation $\bar{\mu}_i=\mu_{L-i+1}$
\bea
U=\{ u_1,\dots,u_{2n+L} \}=\{\mu_1,\dots, \mu_{L},\bar{\gl}_{n},\dots,\bar{\gl}_1,\nu_1,\dots,\nu_{n} \}
\eea
and the full set of rapidities
\bea
V&=&\{ u_1,\dots,u_{2n+L} \}\nn \\
&=&\{\mu_1,\dots, \mu_{L},\bar{\gl}_{n},\dots,\bar{\gl}_1,\nu_1,\dots,\nu_{n},\bar{\mu}_1-\gamma ,\dots,\bar{\mu}_{L}-\gamma, \bar{\nu}_{1}-\gamma,\dots,\bar{\nu}_1-\gamma,\gl_1-\gamma,\dots,\gl_{n}-\gamma\} \nn \\
\eea

We have the out arrows appearing in three distinct blocks.  The positions of the out arrows are
\bea
X=\{ 1,\dots,L, L+n+1, \dots, L+2n,2L+2n+1,\dots,2L+3n \}
\eea
With these specialisations, we can see that a number of terms of the original $(L+2n)!$ permutations will vanish. 

If we first consider the $\phi(w_i,x_i)$ terms, we find that the first $L$ of these (corresponding to the first set of out arrows, associated to the  $\mu$ rapidities), will be zero, except for the term arising from the poles at $w_i =\mu_i$ for $i=1,\dots,L$.  It is then easy to perform the first $L$ integrals, and we find that the integral formula for the scalar product becomes:
\begin{align}\label{intF}
N(\{\nu_i\},\{\lambda_j\}) &= \frac{(-1)^{(L-1)n} (\sinh{\gamma})^{2n} }{(2\pi \ii)^{2n}} \prod_{i,j=1}^{n} \frac{1}{a(\lambda_i - \nu_j) a(\nu_j - \lambda_i) a(\nu_i - \nu_j)^2}   \nn \\
& \times \oint \dots \oint \dd{w}_{L+1} \dots \dd w_{L+2n}  \Bigg( \prod_{i=1}^n F({w}_{L+i},{w}_{L+n+i})  \Bigg) \nn \\
& \times \frac{\prod_{i=1}^{2n}\prod_{j=i+1}^{2n} a({w}_{L+j}-{w}_{L+i}) b({w}_{L+j}-{w}_{L+i})} {\prod_{i=1}^{2n} \prod_{j=1}^{n}  b({w}_{L+i}-\lambda_j) b({w}_{L+i}-\nu_j) } \nn\\
& \times \prod_{i=1}^{n} \prod_{j=1}^{L} a({w}_{L+i}-\mu_j) b(\mu_j - {w}_{L+n+i}) \prod_{i=1}^{n} \prod_{j=1}^{n} a({w}_{L+i}-\nu_j)a(\nu_j - {w}_{L+n+i})\nn\\
&\times \prod_{i=1}^{n} \prod_{j=i+1}^{n} b(\nu_j - {w}_{L+i}) b({w}_{L+2n+1-i} - \nu_j ) a(\nu_i - {w}_{L+j}) a({w}_{L+2n+1-j} - \nu_i ) \; . \nn \\
\end{align}
with
\begin{align}
F({w}_{L+i},{w}_{L+n+i})&=\prod_{j=1}^L \frac{a({w}_{L+i}-\mu_j)a({w}_{L+n+i}-\mu_j) a(\mu_j-{w}_{L+i}) a(\mu_j{w}_{L+n+i})}{(a(\gl_i-\mu_j)) (a(\nu_i-\mu_j)) a(\mu_j-\gl_i) a(\mu_j-\nu_i)} \nn \\
&\quad\times\prod_{j=1}^n \frac{a({w}_{L+i}-\gl_j)a({w}_{L+n+i}-\gl_j)a(\gl_j-{w}_{L+i})a(\gl_j-{w}_{L+n+i})}{(a(\gl_i-\gl_j))^2a(\gl_i-\nu_j)a(\nu_i-\gl_j)}
\end{align}

If we now consider the $\phi(w_i,x_i)$ terms corresponding to the next two sets of out arrows, we find that we have terms that cancel all poles at $\nu_j$ except for $j=i$ for $w_{L+i}$ and that also the only poles associated to the $w_{L+n+i}$ variables remaining in the set of $\bar{\nu}$ rapidities are again those when $i=j$. 

In summary, we have poles at:
\begin{align}\label{eq:poles}
 w_{i}&=\mu_{i} \hspace{3cm} \mbox{ for } i=1,\dots,L, \nn \\
 w_{L+i}&=\nu_{i} \mbox{ or } \gl_{\pi_{ i}} \hspace{2cm}\mbox{ for } i=1,\dots, n,\\
  w_{L+n+i}&=\nu_{n+1-i} \mbox{ or } \gl_{\pi_{\rm{n+1-i}}}  \hspace{6.5mm} \mbox{ for } i=1,\dots, n.\nn
\end{align}
This enables us to further simplify (\ref{intF}).  It is easy to see that the factor $F$ is equal to one for every non-vanishing choice of poles in (\ref{eq:poles}), and thus we obtain the expression (\ref{eq:intsp}).

\providecommand{\newblock}{}

\end{document}